\documentclass[aps,prd,onecolumn,11pt]{revtex4}
\usepackage{xcolor}
\usepackage{amsmath,amssymb,amsthm}
\usepackage{graphicx}
\usepackage{bm}
\usepackage{hyperref}
\usepackage{enumitem}
\usepackage{booktabs}
\usepackage{subcaption} % replacement for subfigure
\usepackage{multirow}
\usepackage{tabularx}
\usepackage{float}
\usepackage{array}
\usepackage{tikz}
\usepackage{natbib} % if you use \citep, \citet, etc.
\usepackage[greek, english]{babel}
\usepackage{alphabeta}
\usepackage{textgreek}
\newtheorem{thm}{Theorem}[section]

\newtheorem{lemma}[thm]{Lemma}
\newtheorem{proposition}[thm]{Proposition}

\usepackage{booktabs}
\usepackage{tcolorbox}

\newcommand{\be}{\begin{equation}}
\newcommand{\ee}{\end{equation}}
\newcommand{\bea}{\begin{eqnarray}}
\newcommand{\eea}{\end{eqnarray}}

\setlist[enumerate,1]{label=(A\arabic*),ref=A\arabic*}
\begin{document}

\title{Dynamics of the Bianchi~V cosmological model inspired by quintessential $\alpha$-attractors}

\author{Genly Leon\textsuperscript{1,2,3}}
\email{genly.leon@ucn.cl}
\author{Amare Abebe\textsuperscript{3,4}}
\email{amare.abebe@nithecs.ac.za}
\author{Andronikos Paliathanasis\textsuperscript{1,2,3,4}}
\email{anpaliat@phys.uoa.gr}
\affiliation{$^{1}$ Departamento de Matem\'aticas, Universidad Cat\'olica del Norte, Avdenida Angamos 0610, Casilla 1280, Antofagasta, Chile, Antofagasta, Chile}
\affiliation{$^{2}$ Institute of Systems Science, Durban University of
Technology, Durban 4000, South Africa}
\affiliation{$^{3}$ Centre for Space Research, North-West University, Potchefstroom 2520, South Africa}
\affiliation{$^{4}$ National Institute for Theoretical and Computational Sciences (NITheCS), South Africa.}

\begin{abstract}
We investigate scalar-field cosmologies in the Bianchi V spacetime using a dynamical-systems framework. Motivated by representative $\alpha$-attractor potentials - the E-model and T-model - we apply averaging theorems and amplitude--phase reductions to monomial potentials $\sim \phi^{2n}$ of the scalar field, which approximate the attractor models near their minima, in the presence of matter with barotropic index $\gamma$. The reduced averaged system admits five generic isolated equilibria: Kasner vacua $\mathcal{K}_0^\pm$, the matter FLRW point $\mathcal{F}$, the scalar FLRW point $\mathcal{S}$, and the curvature Milne-type point $\mathcal{K}$, together with special families for tuned $(n,\gamma)$. We find that $\mathcal{K}_0^\pm$ are always sources, $\mathcal{F}$ is generically a saddle but can act as a sink for $\gamma<\min\{\tfrac{2n}{n+1},\tfrac{2}{3}\}$, $\mathcal{S}$ is a sink if $0<n<\tfrac{1}{2}$ and $\tfrac{2n}{n+1}<\gamma\leq 2$, while $\mathcal{K}$ becomes a sink whenever $\gamma>\tfrac{2}{3}$ and $n>\tfrac{1}{2}$. These results demonstrate that isotropic FLRW $\alpha$-attractor models extend naturally to anisotropic Bianchi~V cosmologies: inflationary attractors remain robust, while the Milne-type curvature solution emerges as the late-time state.
\end{abstract}
\maketitle

\section{Introduction}
The standard cosmological model, commonly referred to as the \textLambda CDM model, assumes that the Universe is homogeneous and isotropic on sufficiently large, cosmological scales, according to the Cosmological Principle. Under these assumptions, the large-scale geometry of spacetime is described by the Friedmann–Lemaître–Robertson–Walker (FLRW) geometry \citep{Weinberg:2008zzc}. The model also assumes the presence of dark components. In particular, cold dark matter is introduced to address numerous gravitational phenomena, such as flat galaxy rotation curves, the growth of large-scale structure, gravitational lensing, and the dynamics of galaxy clusters \citep{Primack:1983xj, Peebles:1984zz, Bond:1984fp, Trimble:1987ee, Turner:1991id,Bertone:2004pz, Frenk:2012ph}. Moreover, the  cosmological constant \textLambda~\citep{Carroll:1991mt} has been introduced to describe the dark energy component of the Universe which drives the late-time cosmic acceleration. The current accelerated phase of the universe is supported by various independent observational probes, including Type Ia supernovae, measurements of the cosmic microwave background, and large-scale structure surveys \citep{SupernovaCosmologyProject:1997zqe, SupernovaSearchTeam:1998fmf, Planck:2018vyg, Percival:2001hw, Eisenstein:2005su, Frieman:2008sn, Weinberg:2012es, Peebles:2002gy}. Despite its minimal set of free parameters, the \textLambda CDM paradigm has proven remarkably successful in providing a coherent and precise description of current cosmological observations \citep{Planck:2018vyg}.

The \textLambda CDM model successfully describes cosmic expansion, structure formation, and the statistical properties of the Universe using a minimal set of assumptions. However, it still faces several serious theoretical and observational challenges. In particular, the fundamental microphysical nature of its dominant components, which account for nearly 95\% of the matter-energy content of the Universe, remains unknown \citep{Weinberg:1988cp, Padmanabhan:2002ji, Bull:2015stt, Bertone:2018krk}. No dark matter particle has been directly detected, and its existence is inferred only from gravitational effects on galaxies, clusters, gravitational lensing, and large-scale structure, whereas dark energy is inferred from its role in driving the late-time accelerated expansion of the Universe \citep{Trimble:1987ee, Turner:1991id, Carroll:1991mt}.

In light of these challenges faced by the General Relativity (GR)–based \textLambda CDM framework, several complementary schools of thought have emerged in recent years. In general, these approaches examine departures from one or more foundational assumptions of the concordance model. Some investigate modifications to the gravitational sector itself, by extending or altering Einstein’s field equations, while others relax the assumptions of large-scale homogeneity and isotropy, or consider more radical changes to the matter–energy content of the universe \citep{Clarkson:2011zq, Buchert:2016sug, Joyce:2014kja, Amendola:2012ys, DiValentino:2021izs, Abdalla:2022yfr}.

In addition to modifying the gravitational sector or introducing new dynamical degrees of freedom, a complementary and physically well-motivated approach to addressing the limitations of the \textLambda CDM paradigm is to relax the assumption of exact isotropy at the level of the background spacetime. Within this broader context, spatially homogeneous but anisotropic Bianchi-type cosmological models provide a natural generalisation of the FLRW framework, enabling the systematic study of anisotropic expansion, shear dynamics, and their impact on cosmic evolution \citep{Ellis:1968vb, Wainwright:1997}. Such models are particularly relevant in light of observational anomalies at large angular scales and the question of whether the observed isotropy of the universe is fundamental or an emergent late-time property \citep{Jaffe:2005pw, Pontzen:2007ii, Rodrigues:2007ny, Bengaly:2016lgr, Orjuela-Quintana:2024}. Among these, Bianchi V spacetimes offer a controlled setting in which deviations from isotropy can be dynamically explored while maintaining close contact with the standard cosmological model. This makes them especially well suited to investigating mechanisms of isotropisation, the late-time suppression of anisotropies, and the role of matter fields—such as scalar fields—in driving the universe towards an effectively FLRW-like state.

In this context, numerous studies have examined how modifying the gravitational field can account for phenomena traditionally attributed to dark components, by extending or altering Einstein’s field equations \citep{Clifton:2011jh, Nojiri:2010wj, DeFelice:2010aj}. Within this broad class of models, scalar fields play a central role and naturally arise in attempts to provide unified descriptions of both the early- and late-time evolution of the universe \citep{Tsujikawa:2013fta, Paliathanasis:2015gga, Kase:2018aps,Urena-Lopez:2011gxx, Leon:2022oyy}. Scalar fields are already a cornerstone of modern cosmology through the inflationary paradigm \citep{Guth:1980zm}, and they also offer well-motivated dynamical alternatives to a pure cosmological constant at late times. These include quintessence models \citep{Ratra:1987rm, Copeland:1999cs}, phantom fields \citep{Nojiri:2005pu}, quintom scenarios \citep{Cai:2009zp, Guo:2004fq}, chiral cosmologies \citep{Dimakis:2020tzc, Chervon:2013btx}, and multi-scalar field frameworks capable of describing successive cosmological epochs within a single theoretical setting \citep{Copeland:1999cs, Tsujikawa:2000wc, Achucarro:2010jv, Akrami:2020zfz}.

From a methodological perspective, scalar-field cosmologies are particularly amenable to qualitative analysis using tools from dynamical systems theory \citep{perko, WE, Coley}, which enable systematic classification of cosmological solutions and rigorous assessment of their stability properties. A well-known example is chaotic inflation, in which the inflationary dynamics are driven by a simple quadratic potential \citep{Linde:1983gd, Linde:1986fd}. Complementary insight can be gained from asymptotic methods and averaging theory \citep{SandersEtAl2010, Fajman:2020yjb,Heinzle:2009zb,Uggla:2003fp}, which are especially effective in probing the structure of the solution space of scalar-field cosmologies, both in vacuum and in the presence of additional matter components \citep{Leon:2019iwj, Leon:2020pfy, Leon:2021lct}. Collectively, these approaches provide a flexible and mathematically robust framework for exploring extensions of the standard cosmological model while remaining closely connected to its phenomenological successes.

One particularly important application of scalar-field cosmology, one that has gained significant traction in recent years \citep{Fajman:2020yjb, Leon:2019iwj, Leon:2020pfy, Leon:2021lct, LeonTorres:2026}, is the construction of a time-averaged description of the underlying cosmological dynamics, in which rapid scalar-field oscillations are systematically smoothed while the physically relevant late-time behaviour is preserved \citep{Fajman:2020yjb,Turner:1983he,Mukhanov:2005sc,Faraoni:2013ig}. This approach is particularly well suited to models governed by the Einstein–Klein–Gordon (EKG) system, in which the scalar field serves as a dynamical source for the spacetime geometry \citep{WE, Coley}. From both a physical and observational perspective, the fine-grained oscillatory dynamics of the scalar field are typically inaccessible, whereas their cumulative, averaged contribution to the energy–momentum tensor directly controls the large-scale expansion and effective matter content of the Universe. Averaging, therefore, provides a principled reduction of the full dynamical system, yielding effective cosmological equations that are simpler to analyse while remaining faithful to the underlying theory \citep{Bruni:2013mua}. Within this framework, one can meaningfully address conceptual and phenomenological issues associated with \textLambda CDM, including the emergence of effective matter components and dynamical alternatives to a rigid cosmological constant.

A particular class of models arises when the scalar field is endowed with a nonlinear or periodic potential, which naturally induces oscillatory behaviour. Such dynamics occur in massive scalar-field models, axion-like scenarios \citep{Hu:2000ke,Marsh:2015xka}, and in the post-inflationary evolution of the early Universe, and can persist over cosmological time scales \citep{Linde:1983gd, Linde:2002ws, fenichel1979, SandersEtAl2010}. Although the field equations are well-posed, the resulting dynamics are intrinsically multiscale \citep{Liddle:1998xm,Dutta:2008fn}: fast oscillations of the scalar field coexist with a slowly evolving cosmological background. Since these microscopic oscillations do not directly affect cosmological observables, their averaged effect governs the effective energy density, pressure, and expansion rate. This clear separation of time scales provides a strong physical motivation for using averaging methods from dynamical systems theory, enabling the construction of effective EKG systems that capture the slow cosmological evolution without explicitly resolving the rapid oscillatory motion.

In scalar-field cosmologies exhibiting oscillatory dynamics, regular oscillations about a stable configuration typically arise only after initial transient phases—such as kinetic- or curvature-dominated evolution—have sufficiently decayed. Once this regime is reached, the subsequent evolution is governed by the asymptotic structure of the underlying dynamical system, which fixes averaged quantities such as the effective equation of state, the dilution rate of the scalar-field energy density, and the long-term behaviour of anisotropies or spatial curvature. From a dynamical systems perspective, this late-time behaviour is commonly governed by attractors, invariant manifolds, or scaling solutions that are largely insensitive to initial conditions. It is in this late-time regime that the averaged description becomes more relevant, as these asymptotic states determine the effective matter content, expansion laws, and stability properties of the cosmological model, thereby providing a robust link between microscopic scalar-field dynamics and macroscopic cosmological phenomenology \citep{Copeland:1997et,Copeland:2006wr,Heard:2002dr,Billyard:1999mv,Kitada:1992uhA,Coley:2004gm,Turner:1983he}.

The paper is organized as follows. In Section \S\ref{sect:II} we introduce the Bianchi V geometry and derive the Einstein–Klein–Gordon equations in Hubble‑normalized variables.  In Section \S\ref{sect:V} we develop the amplitude–phase reduction and averaging analysis for monomial minima, and derive the reduced averaged system.  Numerical simulations and phase portraits are given in Section \S\ref{sect:VI}.  
Finally, in Section \S\ref{sect:VII} we summarize our results and outline directions for future work. 

In Appendix \S\ref{avrg}, we provide a precise formulation of the averaging theorem for monomial scalar potentials, state the technical hypotheses, and prove the amplitude–phase and virial lemmas. Furthermore, we construct the near‑identity averaging transformation with uniform remainder estimates, deduce the long‑time closeness estimate and the coincidence of $\omega$-limit sets, and classify the equilibrium points and relate them to the full system.

\section{Bianchi V}
\label{sect:II}
In this work, we will use the following parametrisation for the Bianchi V metric
spacetimes~\cite[Eq. (2.10)]{Coley:1999uh} and \cite{Millano:2023vny}:
\begin{equation}
ds^2 = -dt^2 + a^2(t)\,dx^2 + b^2(t)\,e^{2x} \left( dy^2 + \frac{a^4(t)}{b^4(t)} dz^2 \right),
\label{ln1}
\end{equation}
with the average Hubble and shear scalar parameters defined as
\begin{equation}
H = \frac{\dot{a}}{a}, \quad \sigma = \frac{\dot{a}}{a} - \frac{\dot{b}}{b}.
\end{equation}
Although the line element (\ref{ln1}) depends on two scale factors, it should not be confused with the locally rotationally symmetric (LRS) Bianchi V spacetime
\begin{equation}
ds^2 = -dt^2 + A^2(t)\,dx^2 + B^2(t)\,e^{2x} \left( dy^2 +  dz^2 \right).
\label{ln2}
\end{equation}
The Bianchi V spacetime belongs to the Class B family of Bianchi models, in which the gravitational field equations are generically non-diagonal \cite{Ellis:1998ct}. In our case, the non-diagonal constraint has already been applied, fixing the third scale factor as a function of the remaining two \cite{Coley:1994bw}, reducing the line element to its present form \eqref{ln1}.

Moreover, the two geometries are different; the LRS Bianchi V supports matter sources with a non-vanishing flux term \cite{King:1972td,Goliath:1998na,Coley:2004jm}, whereas the geometry considered here does not support anisotropic solutions with a cosmic fluid with off-diagonal terms. Consequently, the spacetime \eqref{ln1} cannot be reduced to LRS Bianchi V geometry for any functional form of the scale factors. The two geometries coincide only in the open FLRW spacetime. 

The evolution equations of the scalar field $\phi$ and matter energy density are:
\begin{subequations}\label{eq:noninteracting}
\begin{align}
\ddot{\phi} + 3H\dot{\phi} + V'(\phi) &=0, \\
\dot{\rho}_m + 3\gamma H\rho_m &=0. \label{matter-cons}
\end{align}
\end{subequations}
Here, $\gamma$ is the barotropic equation of state parameter relating the matter density to its pressure $p_m$ given by $p_m=(\gamma-1)\rho_m$. 

The rest of the equations are
\begin{subequations}\label{eq:BV_system}
\begin{align}
\dot{\sigma} &= -3H\sigma, \label{sigma-eq}\\
\dot{a} &= aH, \label{a-eq}\\
\dot{b} &= b(H - \sigma),\label{b-eq}\\
\dot{H} &= -\tfrac{1}{2}\left(\gamma \rho_m + 2\sigma^2 + \dot{\phi}^2\right)- \frac{1}{a^2},\end{align} and the constraint
\begin{align}
3H^2 & = \sigma^2 + \rho_m + \frac{1}{2}{\dot{\phi}}^2+ V(\phi) + \frac{3}{a^2}. \label{Fried-eq}
\end{align}
\end{subequations}

The choice of Bianchi V spacetimes provides a new and nontrivial context for the study of $\alpha$--attractor potentials. Unlike the spatially flat FLRW models, Bianchi V geometries incorporate \emph{negative spatial curvature} together with anisotropic shear degrees of freedom. This combination allows one to probe how inflationary and quintessential dynamics behave in universes that are not exactly isotropic or spatially flat, thereby testing the robustness of attractor behavior under more general conditions and exploring the existence and behaviour of anisotropic solutions, as well as the isotropic limit. 

\subsection{The E--Model Potential}

We start our analysis with the $E$--model potential
\begin{equation}
V(\phi) = V_0\left(1 - e^{-\sqrt{\frac{2}{3\alpha}}\,\phi}\right)^{2n},
\label{eq:Epotential}
\end{equation}
which is nonnegative and features a single global minimum at $(\phi, V) = (0, 0)$, corresponding to a Minkowski vacuum solution $(H, \dot\phi, \phi) = (0, 0, 0)$. The potential exhibits a plateau $V \to V_0$ as $\phi \to +\infty$, while for large negative field values it scales exponentially
\begin{equation}
V(\phi) \sim V_0\,e^{-2n\sqrt{\frac{2}{3\alpha}}\,\phi} \quad \text{as } \phi \to -\infty.
\label{eq:Eexp}
\end{equation}
Near the origin, the potential behaves as a power law
\begin{equation}
V(\phi) \sim \phi^{2n} \quad \text{as } \phi \to 0.
\label{eq:Esmall}
\end{equation}
This asymmetric structure, with a steep exponential tail and a flat inflationary plateau, makes the $E$--model a prototypical example of $\alpha$--attractor inflationary potentials \cite{Alho:2017opd,Kallosh:2013hoa,Kallosh:2013yoa,Galante:2014ifa}.

\subsection{The T--Model Potential}

We now consider the T--model potential \cite{Alho:2017opd,Kallosh:2013yoa}
\begin{equation}
V(\phi) = V_0\,\tanh^{2n}\!\left(\frac{\phi}{\sqrt{6\alpha}}\right),
\label{eq:Tpotential}
\end{equation}
a nonnegative function symmetric under $\phi \mapsto -\phi$, due to the odd parity of the hyperbolic tangent and the even exponent $2n$. The potential attains its global minimum at $\phi = 0$, where $V(0) = 0$, and thus admits a Minkowski vacuum solution at $(H, \dot\phi, \phi) = (0, 0, 0)$.

As $\phi \to \pm\infty$, the hyperbolic tangent saturates to $\pm 1$, and the potential asymptotes to a finite plateau $V(\phi) \to V_0$. 
This yields two symmetric asymptotic de Sitter states, in contrast to the $E$--model, which features a single plateau as $\phi \to +\infty$ and an exponential decay as $\phi \to -\infty$. The double--plateau structure of the T--potential is particularly relevant in cosmological scenarios involving bouncing, cyclic, or emergent dynamics, where the scalar field may traverse both asymptotic regimes.

Near the origin, the T--potential behaves as
\begin{equation}
V(\phi) \sim V_0 \left(\frac{\phi}{\sqrt{6\alpha}}\right)^{2n} \quad \text{as } \phi \to 0,
\label{eq:Tsmall}
\end{equation}
matching the small--field behavior of the $E$--model. This ensures that both potentials support similar early--time dynamics near the Minkowski point, while differing in their global structure and asymptotic behavior.

In summary, the T--potential provides a symmetric, bounded deformation of the standard power--law potential, with a central Minkowski minimum and a double de Sitter plateau. Its compact analytic form makes it well-suited to dynamical systems analysis, particularly when combined with angular compactification techniques that render the state space regular and enable a global classification of asymptotic regimes.

A complementary line of research is due to Alho and Mena~\cite{Alho:2019pku}, who developed a global dynamical systems formulation for flat Robertson--Walker cosmologies with Yang--Mills fields and a perfect fluid. Their framework established rigorous results on the global dynamics, including asymptotic source dominance in both time directions. For the pure massless Yang--Mills case, they embedded explicit solutions into a compact state--space picture, clarifying their global role. This work illustrates the strength of compact dynamical systems methods in extending local analyses to global classifications, and provides a methodological precedent for the compactification and averaging techniques employed here.

An important related contribution is Alho et al.~\cite{Alho:2022qri}, which studied minimally coupled scalar fields with monomial potentials $V(\phi) = (\lambda \phi)^{2n}/(2n)$ interacting with a perfect fluid in flat Robertson--Walker spacetimes. Introducing a friction--like term $\Gamma(\phi) = \mu \phi^{2p}$, they identified a bifurcation at $p=n/2$: for $p<n/2$ the future dynamics admit a variety of attractors, while for $p>n/2$ the asymptotics resemble the non--interacting case. Special cases yield Liénard--type dynamics or new attractors driving late--time acceleration, while $p=n/2$ produces fluid--dominated or oscillatory behaviour. Crucially, they showed that a quasi--de Sitter inflationary solution always exists toward the past, suggesting new realizations of quintessential inflation.

Finally, Alho and Uggla~\cite{Alho:2023pkl} extended the compact dynamical systems framework to quintessential $\alpha$--attractor inflation, unifying early--time inflation with late--time quintessence. In this formulation, the inflationary attractor corresponds to the unstable center manifold of a de Sitter fixed point, while compactification ensures that Minkowski vacua, de Sitter plateaus, and exponential tails are consistently incorporated. Their results demonstrate that the $E$-- and $T$--potentials embed naturally into the quintessential $\alpha$--attractor paradigm.

Based on the above discussion, we proceed with the averaging strategy in the following sections.

\section{Averaging near the minimum of the potentials}
\label{sect:V}

Rescaling $V_0$, the $E$-- and $T$--potentials reduce near $\phi=0$ to the monomial scalar field potential
\begin{equation}\label{eq:V_monomial}
V(\phi) = \frac{\mu^{2n}}{2n}\,\phi^{2n}, \quad n>0.
\end{equation}
The corresponding Klein--Gordon equation in the Bianchi~V spacetime is
\begin{align}
\ddot{\phi} + 3H\dot{\phi} + \mu^{2n}\phi^{2n-1} &= 0. \label{eq:KG}
\end{align}
The matter conservation law, anisotropy evolution, and scale-factor dynamics follow from Eqs.~\eqref{matter-cons}--\eqref{b-eq} and the constraint is \eqref{Fried-eq}.

\subsection{Amplitude--phase representation}
\label{VA}

Near the minimum of the monomial potential, the field oscillates rapidly. To capture this dynamics, we introduce
\begin{equation}\label{eq:phi_AP}
\phi(t) = A(t)\cos\theta(t),
\end{equation}
with slowly varying amplitude $A(t)$ and rapidly varying phase $\theta(t)$. Differentiating gives
\begin{align}
\dot{\phi} &= -A\dot{\theta}\sin\theta + \dot{A}\cos\theta, \label{eq:phidot_AP}\\
\ddot{\phi} &= -A(\dot{\theta})^2\cos\theta - A\ddot{\theta}\sin\theta - 2\dot{A}\dot{\theta}\sin\theta + \ddot{A}\cos\theta. \label{eq:phiddot_AP}
\end{align}

Substituting into \eqref{eq:KG} yields
\begin{align}
&\ddot{A}\cos\theta - A(\dot{\theta})^2\cos\theta - A\ddot{\theta}\sin\theta - 2\dot{A}\dot{\theta}\sin\theta \nonumber\\
&\quad + 3H(\dot{A}\cos\theta - A\dot{\theta}\sin\theta) + \mu^{2n}A^{2n-1}\cos^{2n-1}\theta = 0. \label{eq:KG_AP}
\end{align}

Introduce auxiliary variables
\begin{equation}\label{eq:aux_vars}
B \equiv \dot{A}, \qquad \Theta \equiv \dot{\theta}, \qquad \dot{B} = \ddot{A}, \qquad \dot{\Theta} = \ddot{\theta}.
\end{equation}
Equation \eqref{eq:KG_AP} becomes
\begin{align}
(\dot{B} - A\Theta^2 + 3HB + \mu^{2n}A^{2n-1}\cos^{2n-2}\theta)\cos\theta
+ (-A\dot{\Theta} - 2B\Theta - 3HA\Theta)\sin\theta = 0.
\end{align}

Since $\cos\theta$ and $\sin\theta$ are linearly independent, the coefficients must vanish separately:
\begin{equation}
\dot{B} = A\Theta^2 - 3HB - \mu^{2n}A^{2n-1}\cos^{2n-2}\theta,
\qquad
\dot{\Theta} = -\frac{2B\Theta}{A} - 3H\Theta,
\end{equation}
together with $\dot{A}=B$ and $\dot{\theta}=\Theta$.

\subsection{Exact amplitude--phase system}

Coupling to gravity yields the full first‑order system
\begin{subequations}\label{eq:AP_system}
\begin{align}
\dot{A} &= B, \label{eq:Aeq}\\
\dot{B} &= A\,\Theta^2 - 3H\,B - \mu^{2n}A^{2n-1}\cos^{2n-2}\theta, \label{eq:Beq}\\
\dot{\theta} &= \Theta, \label{eq:thetaeq}\\
\dot{\Theta} &= -\frac{2B\Theta}{A} - 3H\,\Theta, \label{eq:Thetaeq}\\
\dot{\Sigma} &= -(2-q)\,\Sigma H, \label{eq:Sigmaprime}\\
\dot{\Omega}_m &= \Omega_m\,(2q - 3\gamma + 2) H, \label{eq:Omdot}\\
\dot{H} &= -(1+q) H^2, \label{eq:Hdot}
\end{align}
\end{subequations}
with deceleration parameter
\begin{equation}\label{eq:q_AP}
q = 2\Sigma^2 + \tfrac{1}{2}(3\gamma-2)\Omega_m 
+ \frac{1}{3H^2}\left[\big(-A\Theta\sin\theta + B\cos\theta\big)^2
- \frac{\mu^{2n}}{2n}(A\cos\theta)^{2n}\right].
\end{equation}

The normalized scalar fraction, expressed in terms of the amplitude--phase variables, is
\begin{equation}\label{eq:Omega_phi_AP}
\Omega_\phi(A,B,\theta,\Theta)
= \frac{1}{3H^2}\left[\tfrac{1}{2}\big(-A\Theta\sin\theta + B\cos\theta\big)^2
+ \frac{\mu^{2n}}{2n}(A\cos\theta)^{2n}\right].
\end{equation}

The only division by $A$ occurs in $\dot{\Theta}$. Since $A$ is the oscillation amplitude, it can be chosen nonnegative and does not generically cross zero for continuous oscillatory solutions. If $A$ becomes very small in numerical runs, one may switch to energy--angle variables or absorb the sign of $A$ into $\theta$ to enforce $A\ge0$. Thus, the amplitude--phase system remains regular across $\phi=0$ and avoids the $\phi^{-1}$ singularity.

Averaging is a coarse-graining procedure in which rapidly oscillating quantities are replaced by their mean values over one period of the fast variable. This yields an effective dynamical system for the slowly varying amplitudes and phases. In the fast--oscillation regime, averaging extracts the secular evolution while filtering out cycle--to--cycle fluctuations, providing a simplified yet accurate description of the long-term dynamics.

\subsection{Averaged amplitude--phase evolution}

Averaging the oscillatory terms in \eqref{eq:AP_system} over one fast cycle of $\theta$ yields
\begin{equation}\label{eq-151}
\overline{\cos^{2n}\theta} = \frac{(2n)!}{2^{2n}(n!)^2}, 
\qquad
\overline{\cos^{2n-2}\theta} = \frac{(2n-2)!}{2^{2n-2}[(n-1)!]^2}.
\end{equation}

The averaged amplitude--phase system is then
\begin{subequations}\label{eq:AP_system_avg}
\begin{align}
\dot{\overline{A}} &= \overline{B}, \\[4pt]
\dot{\overline{B}} &=  \overline{A}\,\overline{\Theta}^2 - 3H\,\overline{B}
- \mu^{2n}\,\overline{A}^{2n-1}\,
\frac{(2n-2)!}{2^{\,2n-2}\,\big[(n-1)!\big]^2}, \\[4pt]
\dot{\overline{\theta}} &= \overline{\Theta}, \\[4pt]
\dot{\overline{\Theta}} &= -\frac{2\overline{B}\,\overline{\Theta}}{\overline{A}} - 3H\,\overline{\Theta}, \\[4pt]
\dot{\overline{\Sigma}} &= -(2-\overline{q})\,\overline{\Sigma}\,H, \\[4pt]
\dot{\overline{\Omega}}_m &= \overline{\Omega}_m\,(2\overline{q} - 3\gamma + 2)\,H, \\[4pt]
\dot{H} &= -(1+\overline{q})\,H^2,
\end{align}
\end{subequations}
with the averaged deceleration parameter
\begin{equation}\label{eq:q_AP_avg}
\overline{q} = 2\overline{\Sigma}^2 + \tfrac{1}{2}(3\gamma-2)\,\overline{\Omega}_m 
+ \tfrac{1}{2}(1+3\overline{w}_\phi)\,\overline{\Omega}_\phi,
\end{equation}
where $\overline{\Omega}_\phi$ denotes the effective scalar contribution.

Averaging \eqref{eq:Omega_phi_AP} over one cycle gives
\begin{equation}
\langle \sin^2\theta \rangle = \langle \cos^2\theta \rangle = \tfrac{1}{2}, 
\qquad \langle \sin\theta\cos\theta \rangle = 0,
\end{equation}
so the averaged scalar fraction becomes
\begin{equation}\label{eq:Omega_phi_avg}
\overline{\Omega}_\phi
= \frac{1}{3H^2}\left[
\tfrac{1}{4}\left(\overline{A}^{2}\,\overline{\Theta}^{2} + \overline{B}^{2}\right)
+ \frac{\mu^{2n}}{2n}\,\overline{A}^{2n}\,\frac{(2n)!}{2^{2n}(n!)^2}
\right].
\end{equation}

\subsection{Averaging procedure for rapid oscillations}

When $\omega(A)\gg H$, the scalar field oscillates much faster than the Hubble expansion, so one can average over a period of $\theta$.

\subsubsection{Frequency}
From \eqref{eq:KG_AP}, the leading balance is
\begin{equation}
-\,A(\dot{\theta})^2\cos\theta \;\sim\; \mu^{2n}A^{2n-1}\cos^{2n-1}\theta. \label{eq:balance}
\end{equation}
Dividing by $A\cos\theta$ gives
\begin{equation}
(\dot{\theta})^2 \approx \mu^{2n}A^{2n-2}\cos^{2n-2}\theta. \label{eq:theta_sq}
\end{equation}
Averaging $\cos^{2n-2}\theta$ yields
\begin{equation}
(\dot{\theta})^2 \cdot \tfrac{1}{2} = \mu^{2n}A^{2n-2}\frac{(2n)!}{2^{2n}(n!)^2}, \label{eq:theta_avg}
\end{equation}
so the effective frequency is
\begin{equation}
 \dot{\theta} = \omega(A) := \mu^n A^{\,n-1}\left(\frac{(2n)!}{2^{2n-1}(n!)^2}\right)^{1/2}.   \label{eq:omega}
\end{equation}

\subsubsection{Amplitude evolution}
The friction term
\begin{equation}
3H\dot{\phi} = 3H(-A\dot{\theta}\sin\theta + \dot{A}\cos\theta) \label{eq:friction}
\end{equation}
averages to a secular drift in $A(t)$. Comparing $\overline{\rho_\phi}\propto A^{2n}$ with the averaged loss
\begin{equation}
-3H\overline{\dot{\phi}^2}= -\tfrac{3n}{n+1}H\,\overline{\rho_\phi}, \label{eq:loss}
\end{equation}
gives
\begin{equation}
\dot{A} = -\tfrac{3H}{n+1}\,A. \label{eq:A_evol}
\end{equation}

\subsubsection{Averaged dynamics}
Equations \eqref{eq:omega} and \eqref{eq:A_evol} define the averaged system, valid when $\omega(A)\gg H$.

\subsubsection{Scaling laws and effective fluid}
Integration of \eqref{eq:A_evol} gives
\begin{equation}
A(t)\propto a^{-\tfrac{3}{n+1}}, \qquad \overline{\rho_\phi}\propto a^{-\tfrac{6n}{n+1}}. \label{eq:scaling}
\end{equation}
The virial relation $\overline{\dot{\phi}^2}=n\,\overline{V(\phi)}$ implies
\begin{equation}
\overline{p_\phi}=\tfrac{n-1}{n+1}\,\overline{\rho_\phi}, \quad
\overline{w_\phi}=\tfrac{n-1}{n+1}, \quad
\overline{\gamma}_\phi=\tfrac{2n}{n+1}. \label{eq:eos}
\end{equation}
Thus, the scalar behaves as a fluid with barotropic index $\overline{\gamma}_\phi$ and dilution law from \eqref{eq:scaling}.

\subsubsection{Exact vs. averaged evolution}
The exact law is
\begin{equation}\label{eq:Omega_phi_original}
\Omega_\phi' \;=\; -6x^2 + 2(q+1)\,\Omega_\phi,
\end{equation}
with
\begin{equation}
\Omega_{\phi}\;:= x^2 + y^2,\quad x \;:=\; \frac{\dot\phi}{\sqrt{6}\,H}, \quad  y \;:=\; \frac{\mu^{n}\phi^{n}}{\sqrt{6n}\,H}, \label{eq:defs}
\end{equation}
and
\begin{equation}\label{eq:q_def}
q \;=\; 2\Sigma^2 + \tfrac{1}{2}(3\gamma-2)\,\Omega_m + 2x^2 - y^2.
\end{equation}
In the rapid oscillation regime, the virial relation gives
\begin{equation}
\overline{x^2} \;=\; \tfrac{n}{n+1}\,\overline{\Omega}_\phi, \label{eq:x_avg}
\end{equation}
so \eqref{eq:Omega_phi_original} averages to
\begin{equation}\label{eq:Omega_phi_evol}
\overline{\Omega}_\phi' \;=\; \overline{\Omega}_\phi\!\left(2\overline{q} - 3\frac{2n}{n+1} + 2\right),
\end{equation}
with constraint
\begin{equation}\label{eq:constraint}
\overline{\Sigma}^2 + \overline{\Omega}_m + \overline{\Omega}_k + \overline{\Omega}_\phi = 1.
\end{equation}

\subsection{Remarks}
Amplitude–phase or action–angle variables \cite{Millano:2023vny,Millano:2025vjo,Leon:2026vov} remove the $\phi^{-1}$ singularity and ensure regularity across oscillations while preserving numerical stability.

\subsection{Connection to the averaged system}
\label{V H}

Averaging \eqref{eq:AP_system} over one fast oscillation period reproduces the slow amplitude law \eqref{eq:A_evol}, consistent with the fluid description \eqref{eq:eos}. Hence, the exact amplitude--phase system is asymptotically consistent with the averaged dynamics and recovers the late‑time barotropic index $\overline{w}_\phi=(n-1)/(n+1)$ for the potential index $n > 0$. The resulting system is
\begin{subequations}\label{eq:averaged_system_reduced}
\begin{align}
\overline{\Sigma}' &= \left[2(\overline{\Sigma}^2-1) + \tfrac{1}{2}(3\gamma-2)\,\overline{\Omega}_m 
+ \frac{(2n-1)}{(n+1)}\,\overline{\Omega}_\phi\right]\overline{\Sigma}, \\
\overline{\Omega}_m' &= \left[4\overline{\Sigma}^2 + (3\gamma-2)\,\left(\overline{\Omega}_m -1\right)+ \frac{2(2n-1)}{(n+1)}\,\overline{\Omega}_\phi\right]\overline{\Omega}_m, \\
\overline{\Omega}_\phi' &= \left[4\overline{\Sigma}^2 + (3\gamma-2)\,\overline{\Omega}_m 
+ \frac{2(2n-1)}{(n+1)}\,(\overline{\Omega}_\phi - 1)\right]\overline{\Omega}_\phi.
\end{align}
\end{subequations}
with phase space defined by the compact constraint
\begin{equation}\label{eq:constraint_reduced_explicit}
\overline{\Sigma}^2 + \overline{\Omega}_m + \overline{\Omega}_\phi \leq 1.
\end{equation}

Equilibrium states are obtained.  We have isolated points (generic $n,\gamma$):
\begin{align*}
(\overline{\Sigma},\overline{\Omega}_m,\overline{\Omega}_\phi) &= (-1,0,0) \quad \text{Kasner vacuum } \mathcal{K}_0^{-}, \\
(\overline{\Sigma},\overline{\Omega}_m,\overline{\Omega}_\phi) &= (0,0,1) \quad \text{Scalar FLRW } \mathcal{S}, \\
(\overline{\Sigma},\overline{\Omega}_m,\overline{\Omega}_\phi) &= (0,1,0) \quad \text{Matter FLRW } \mathcal{F}, \\
(\overline{\Sigma},\overline{\Omega}_m,\overline{\Omega}_\phi) &= (1,0,0) \quad \text{Kasner vacuum } \mathcal{K}_0^{+}, \\
(\overline{\Sigma},\overline{\Omega}_m,\overline{\Omega}_\phi) &= (0,0,0) \quad \text{Curvature point } \mathcal{K}.
\end{align*}
For certain $(n,\gamma)$, the system admits one‑dimensional families:
\begin{itemize}
    \item $n=\tfrac{1}{2}, \gamma=\tfrac{2}{3}$: line $\overline{\Sigma}=0$ with arbitrary $(\overline{\Omega}_m,\overline{\Omega}_\phi)$ subject to \eqref{eq:constraint_reduced_explicit}.
    \item $n=\tfrac{1}{2}$: curve $\overline{\Sigma}=0$, $\overline{\Omega}_m=0$, arbitrary $\overline{\Omega}_\phi$.
    \item $\gamma=\tfrac{2}{3}$: curve $\overline{\Sigma}=0$, $\overline{\Omega}_\phi=0$, arbitrary $\overline{\Omega}_m$.
    \item $\gamma=\tfrac{2n}{1+n}$: mixed scalar–matter scaling line $\overline{\Sigma}=0$, $\overline{\Omega}_\phi=1-\overline{\Omega}_m$, interpolating between $\mathcal{F}$ and $\mathcal{S}$.
    \item $\gamma=2$: curve $\overline{\Omega}_m=1-\overline{\Sigma}^2$, $\overline{\Omega}_\phi=0$, a family of matter–curvature states.
\end{itemize}

The reduced averaged system exhibits both isolated equilibria and continuous families of solutions that interpolate between the Kasner vacua, the matter FLRW point, the scalar FLRW point, and the curvature point. As summarized in Table~\ref{tab:eigenvalues_generic}, these five generic equilibria---$\mathcal{K}_0^{\pm}$, $\mathcal{F}$, $\mathcal{S}$, and $\mathcal{K}$---possess eigenvalue spectra and stability properties that depend sensitively on the parameters $(n,\gamma)$, giving rise to sources, saddles, and sinks in different regimes. The stability conditions in Table~\ref{tab:eigenvalues_generic} further demonstrate the coexistence of multiple attractors, with both $\mathcal{S}$ and $\mathcal{K}$ capable of acting as sinks depending on $(n,\gamma)$, thereby shaping the global phase flow of the system.

\begin{table}[ht]
\centering
\caption{Equilibria and eigenvalues of the averaged system \eqref{eq:averaged_system_reduced} 
for general $n>0$ and $0\leq\gamma\leq 2$.}
\begin{tabular}{l c l}
\hline
Equilibrium & Eigenvalues & Stability \\ \hline

$\mathcal{K}_0^{\pm}$ at $(\pm 1,0,0)$ 
& $\left\{ \tfrac{6}{1+n}, \; 4, \; 6-3\gamma \right\}$ 
& Source for $0\leq \gamma <2$; nonhyperbolic at $\gamma=2$ \\[6pt]

$\mathcal{F}$ at $(0,1,0)$ 
& $\left\{ \tfrac{3}{2}(\gamma-2), \; -2+3\gamma, \; -\tfrac{6n}{1+n}+3\gamma \right\}$ 
& Sink if $\gamma<\min\{\tfrac{2n}{n+1},\tfrac{2}{3}\}$; 
nonhyperbolic at $\gamma=\tfrac{2}{3}, \tfrac{2n}{n+1},$ or $2$; \\
& & saddle otherwise \\[6pt]

$\mathcal{S}$ at $(0,0,1)$ 
& $\left\{ -\tfrac{3}{1+n}, \; 4-\tfrac{6}{1+n}, \; \tfrac{6n}{1+n}-3\gamma \right\}$ 
& Sink if $0<n<\tfrac{1}{2}, \, \tfrac{2n}{n+1}<\gamma\leq 2$; \\
& & nonhyperbolic at $n=\tfrac{1}{2}$ or $n=\tfrac{\gamma}{2-\gamma}$; 
saddle otherwise \\[6pt]

$\mathcal{K}$ at $(0,0,0)$ 
& $\left\{ -2, \; -4+\tfrac{6}{1+n}, \; 2-3\gamma \right\}$ 
& Sink if $\gamma>\tfrac{2}{3}$ and $n>\tfrac{1}{2}$; 
nonhyperbolic at $n=\tfrac{1}{2}$ or $\gamma=\tfrac{2}{3}$; \\
& & saddle otherwise \\[6pt]
\hline
\end{tabular}
\label{tab:eigenvalues_generic}
\end{table}

\section{Numerical implementation}
\label{sect:VI}
In practice, we integrate both the full amplitude--phase system~\eqref{eq:AP_system} 
and the averaged reduced system~\eqref{eq:eos} in cosmic time~$t$. 
The state vector is defined as $(A,B,\theta,\Theta,\Sigma,\Omega_m,H)$,
with $\phi$ and $\dot{\phi}$ reconstructed via~\eqref{eq:phi_AP}, and 
~\eqref{eq:phidot_AP}. 
The averaged system evolves $(\Omega_\phi,\Omega_m,\Sigma,H)$, thereby smoothing over 
the fast oscillations. Both systems are solved using a stiff integrator 
(\texttt{Radau}) with tight tolerances over the interval $t \in [0,T_{\text{max}}]$, 
with $T_{\text{max}} = 40,\,200$. Short integration times highlight the rapid oscillations, and long times reveal the capture of the transitionary dynamics, the late--time behavior. At each step, $\Omega_\phi$ is 
computed from~\eqref{eq:Omega_phi_AP} and used to initialize the averaged system, 
ensuring consistent initial conditions. The resulting trajectories are then 
compared in the three--dimensional phase space $(\Omega_\phi,\Sigma,\Omega_m)$ 
and its two--dimensional projections.

\begin{table}[h!]
\centering
\caption{Initial conditions for the full amplitude--phase system with $n=2, 3$. 
The colored stars show the reference colors used in the plots.}
\label{tb_1}
\begin{tabular}{c c c c c c c c}
\hline
Set & $A(0)$ & $B(0)$ & $\theta(0)$ & $\Theta(0)$ & $\Sigma(0)$ & $\Omega_m(0)$ & $H(0)$ \\
\hline
\textcolor{blue}{$\star$}~1 & 1.0 & 0.1 & 0.05 & 1.0 & 0.12 & 0.08 & 1.0 \\
\textcolor{orange}{$\star$}~2 & 1.2 & 0.1 & 0.05 & 1.0 & 0.18 & 0.05 & 1.0 \\
\textcolor{green}{$\star$}~3 & 1.5 & 0.1 & 0.05 & 1.0 & 0.22 & 0.12 & 1.0 \\
\textcolor{red}{$\star$}~4 & 1.1 & 0.1 & 0.05 & 1.0 & 0.14 & 0.20 & 1.0 \\
\textcolor{purple}{$\star$}~5 & 1.3 & 0.1 & 0.05 & 1.0 & 0.20 & 0.15 & 1.0 \\
\hline
\end{tabular}
\end{table}

These initial conditions explore different balances among scalar-field energy, shear, and matter fractions. By computing $\Omega_\phi$ from \eqref{eq:Omega_phi_AP} at $t=0$ and feeding it into the averaged system, both descriptions start consistently, allowing a direct comparison of how the averaged dynamics track the full oscillatory behaviour.

\subsection{Dynamics in the coordinates $(A,B,\theta,\Theta,\Sigma,\Omega_m,H)$}

Figures \ref{fig:time_n1} and \ref{fig:time_n2} illustrate the accuracy of the averaged amplitude--phase system in reproducing the slow dynamics of the full model. For $n=1$ and $n=2$, the averaged trajectories capture the envelope of the oscillations in $A(t)$, $B(t)$, and $\Theta(t)$, while eliminating the fast phase dependence on $\theta$. The shear $\Sigma(t)$ and matter density parameter $\Omega_m(t)$ show close agreement between the two descriptions, confirming that the averaged system preserves the essential gravitational backreaction. Most importantly, the scalar contribution $\Omega_\phi(t)$ is well tracked by the averaged system, validating the effective fluid description with barotropic index $w_\phi=(n-1)/(n+1)$.

\begin{figure}[ht!]
    \centering
    \includegraphics[width=0.7\textwidth]{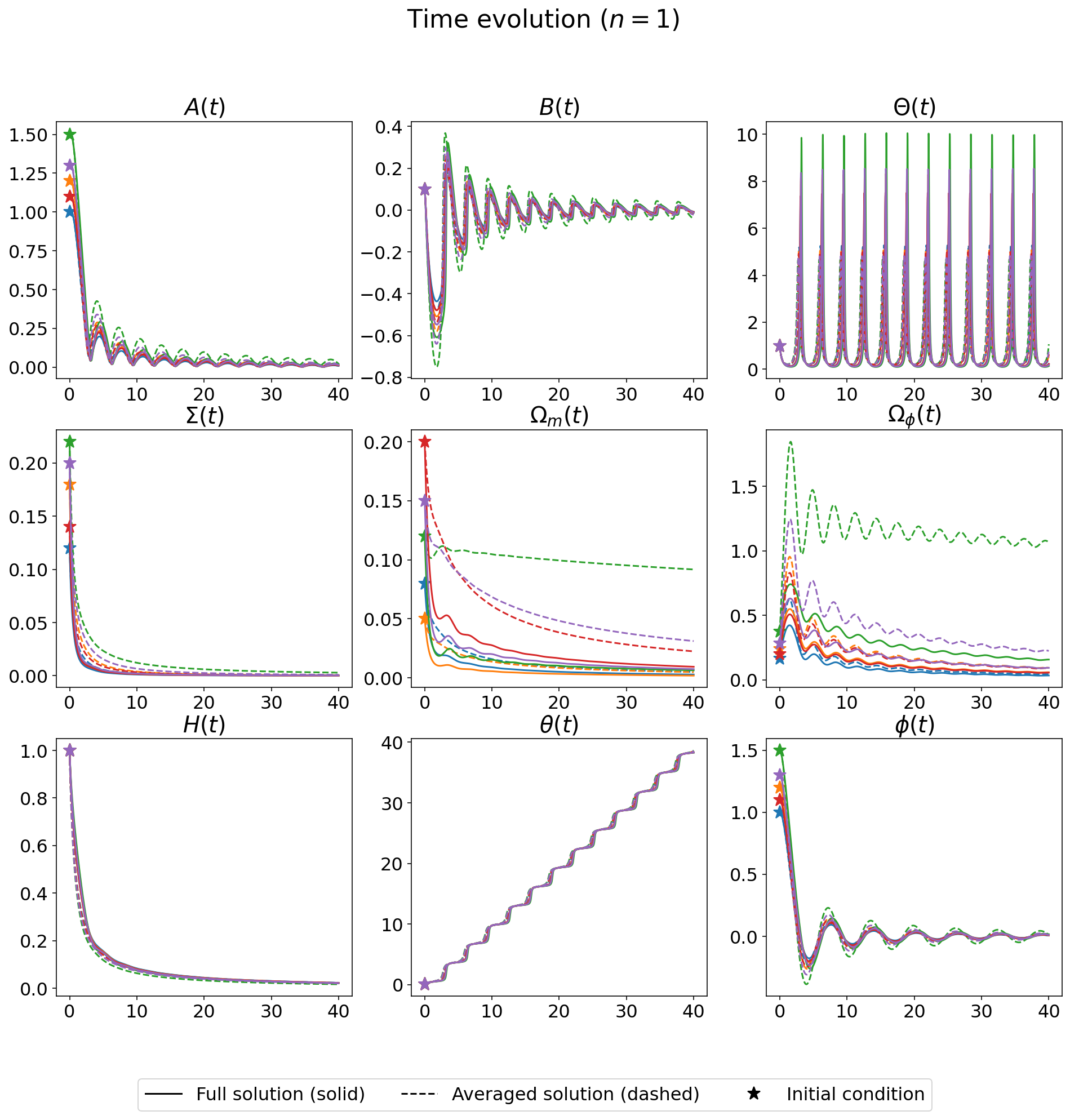}
    \caption{Time evolution of the dynamical variables for $n=1$. 
    Solid lines correspond to the averaged system, while dashed lines correspond to the full system. Initial conditions are taken from Table \ref{tb_1}.}
    \label{fig:time_n1}
\end{figure}

\begin{figure}[ht!]
    \centering
    \includegraphics[width=0.7\textwidth]{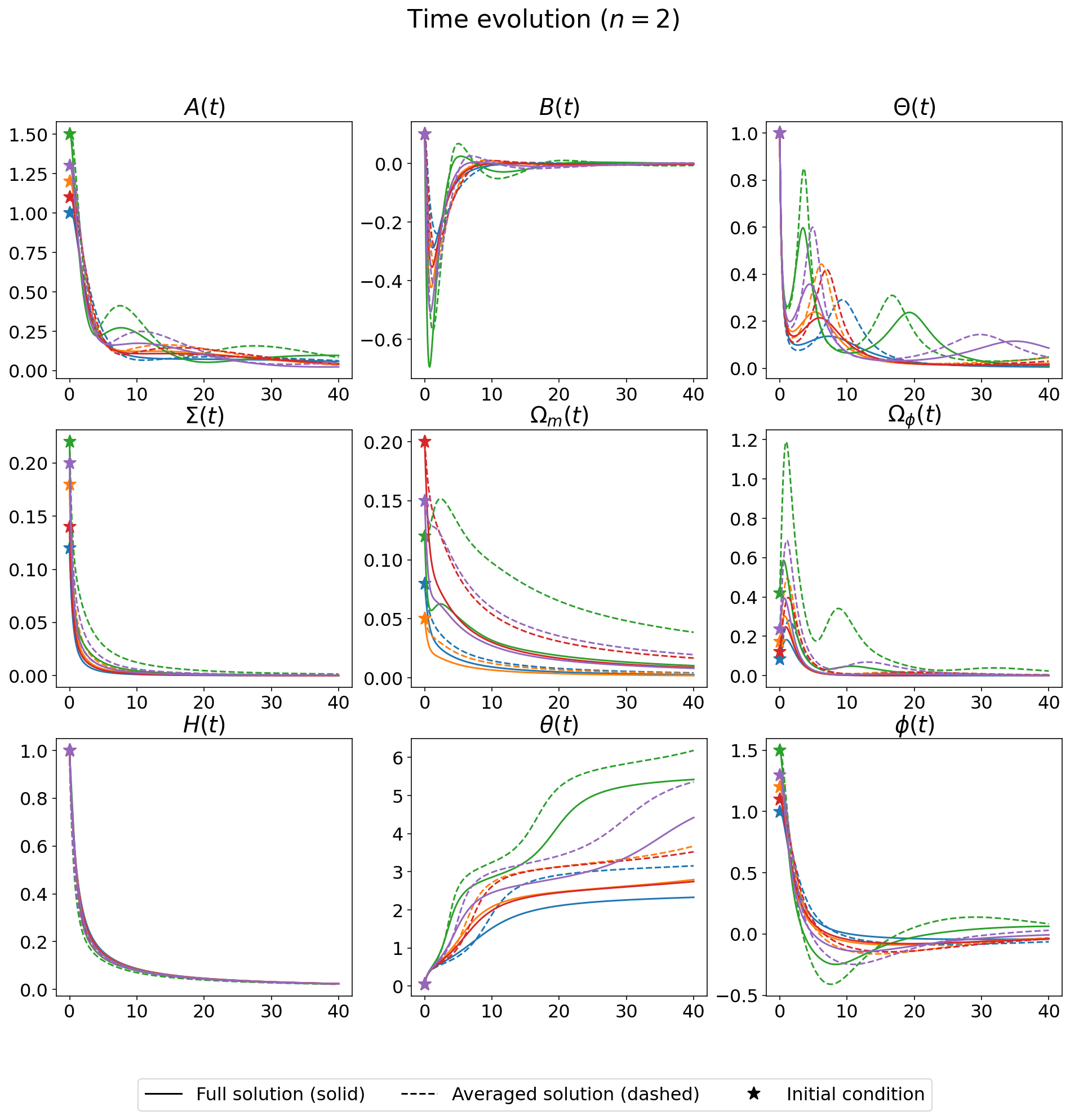}
    \caption{Time evolution of the dynamical variables for $n=2$. 
    Solid lines correspond to the averaged system, while dashed lines correspond to the full system. Initial conditions are taken from Table \ref{tb_1}.}
    \label{fig:time_n2}
\end{figure}

Figures \ref{fig:phase2d_n1} and \ref{fig:phase2d_n2} present the corresponding 2D phase--space portraits. The projections (including $A$ vs.\ $B$, $\theta$ vs.\ $\Theta$, $B$ vs.\ $\theta$, $B$ vs.\ $\Theta$, $\phi$ vs.\ $\dot{\phi}$, $\Sigma$ vs.\ $H$, $\Omega_m$ vs.\ $H$, $\Omega_\phi$ vs.\ $H$, $\Sigma$ vs.\ $\Omega_m$, $\phi$ vs.\ $H$, and $\phi$ vs.\ $\Sigma$) demonstrate that the averaged system reproduces the slow manifold structure of the full dynamics. The averaged trajectories smooth out the fast oscillations in $\theta$, while the full system shows the underlying oscillatory modulation. Both descriptions converge toward the same attractor manifold.

\begin{figure}[ht!]
    \centering
    \includegraphics[width=0.7\textwidth]{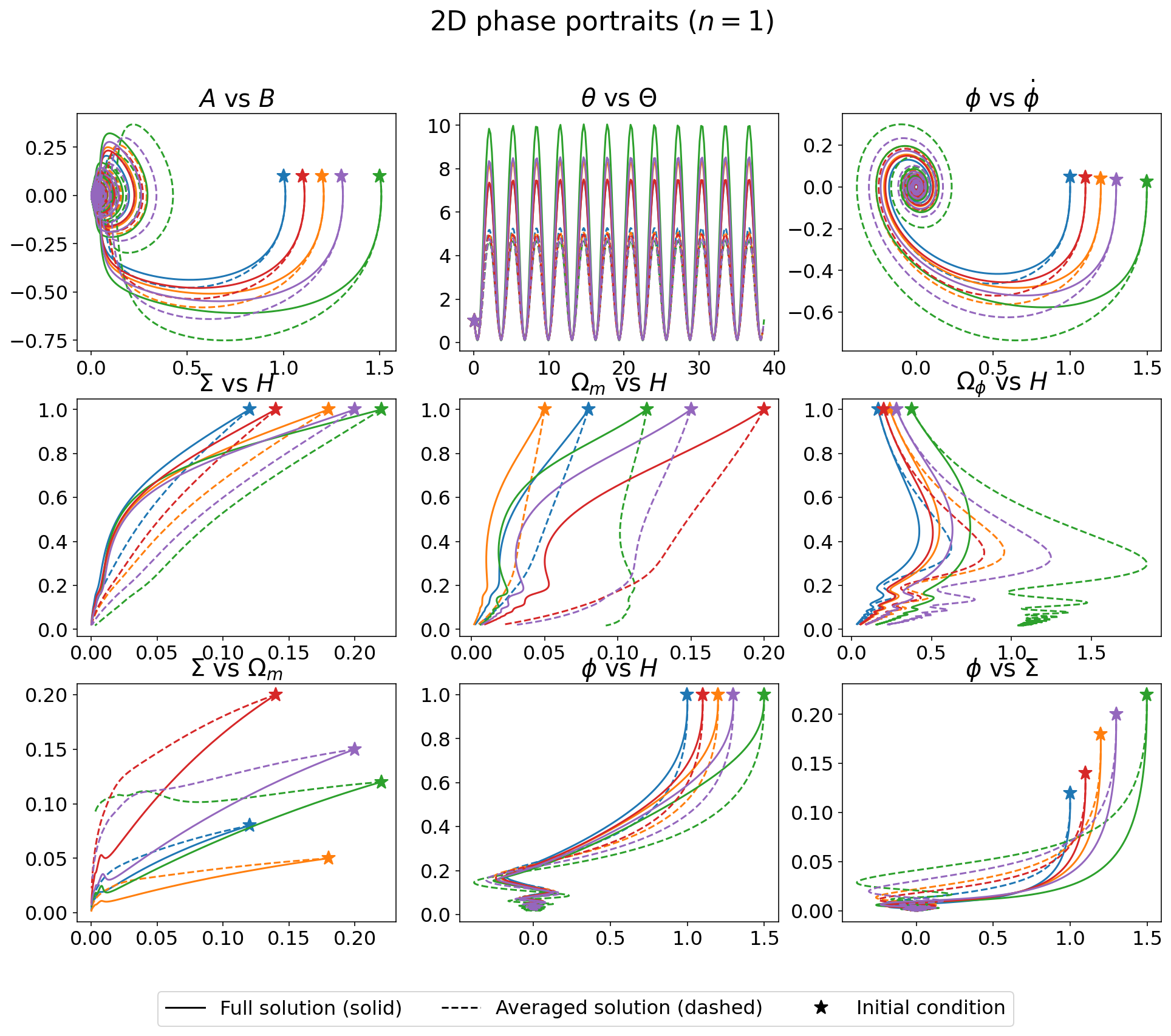}
    \caption{2D phase--space portraits for $n=1$. 
    Projections include $A$ vs.\ $B$, $\theta$ vs.\ $\Theta$, $B$ vs.\ $\theta$, $B$ vs.\ $\Theta$, $\phi$ vs.\ $\dot{\phi}$, and $\phi$ vs.\ $H$. 
    The averaged system reproduces the slow manifold structure of the full dynamics. Initial conditions are taken from Table \ref{tb_1}.}
    \label{fig:phase2d_n1}
\end{figure}

\begin{figure}[ht!]
    \centering
    \includegraphics[width=0.7\textwidth]{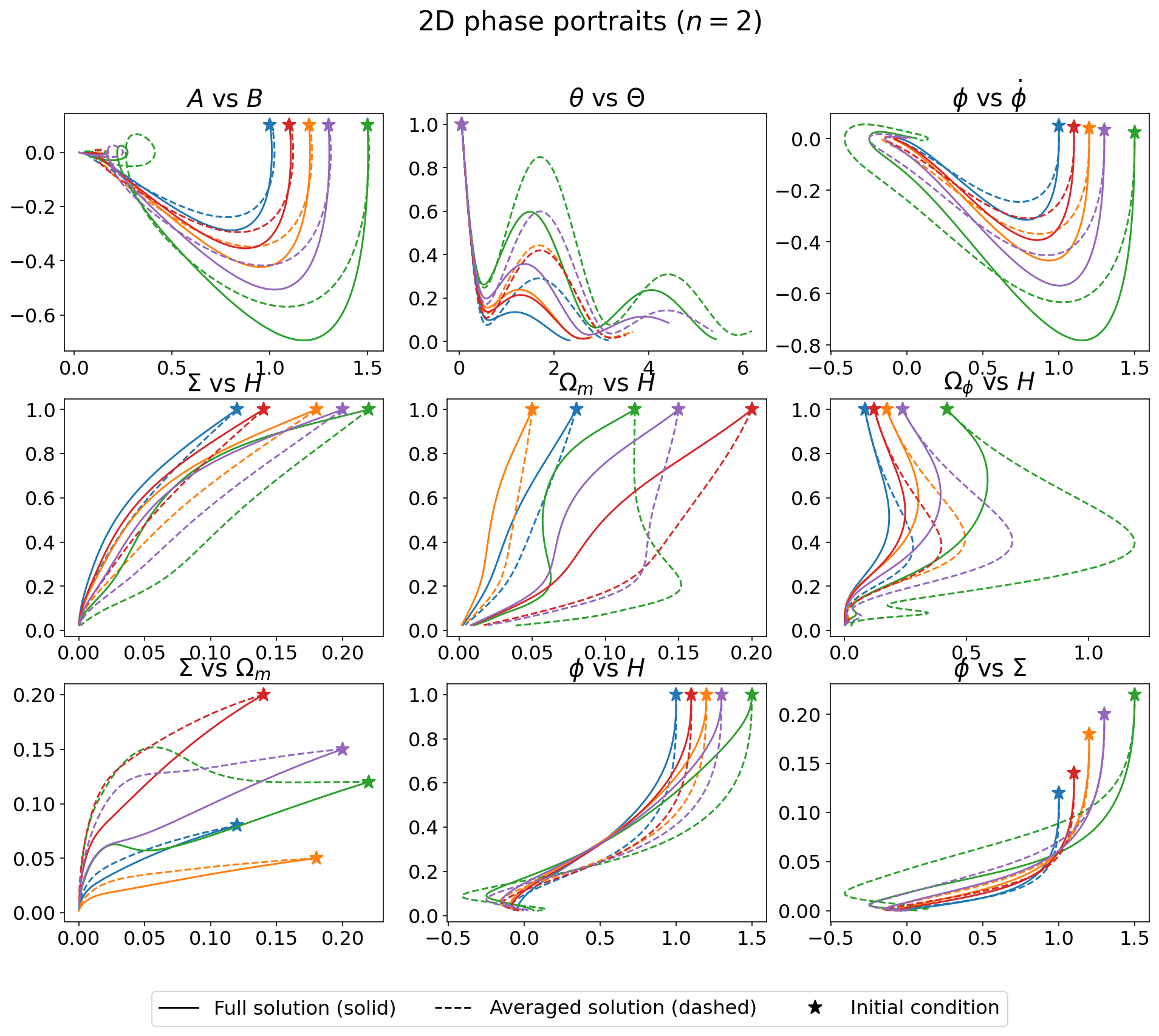}
    \caption{2D phase--space portraits for $n=2$. 
    Projections include $A$ vs.\ $B$, $\theta$ vs.\ $\Theta$, $B$ vs.\ $\theta$, $B$ vs.\ $\Theta$, $\phi$ vs.\ $\dot{\phi}$, and $\phi$ vs.\ $H$. 
    The averaged system reproduces the slow manifold structure of the full dynamics. Initial conditions are taken from Table \ref{tb_1}.}
    \label{fig:phase2d_n2}
\end{figure}

Figures \ref{fig:phase3d_n1} and \ref{fig:phase3d_n2}, show the 3D phase portraits. The plots in $(\phi,\dot{\phi},H)$ and $(\phi,\dot{\phi},A)$ indicate how the scalar field dynamics couple to the Hubble parameter and amplitude, while the portraits in $(\Omega_m,H,\Sigma)$ and $(\Omega_\phi,H,\Sigma)$ to show the relation between matter, scalar energy, and shear. Small deviations are observed in the early-time oscillatory regime, but they diminish as the system evolves, which means that the averaged equations provide a reliable autonomous approximation to the long-term dynamics.

\begin{figure}[ht!]
    \centering
    \includegraphics[width=0.8\textwidth]{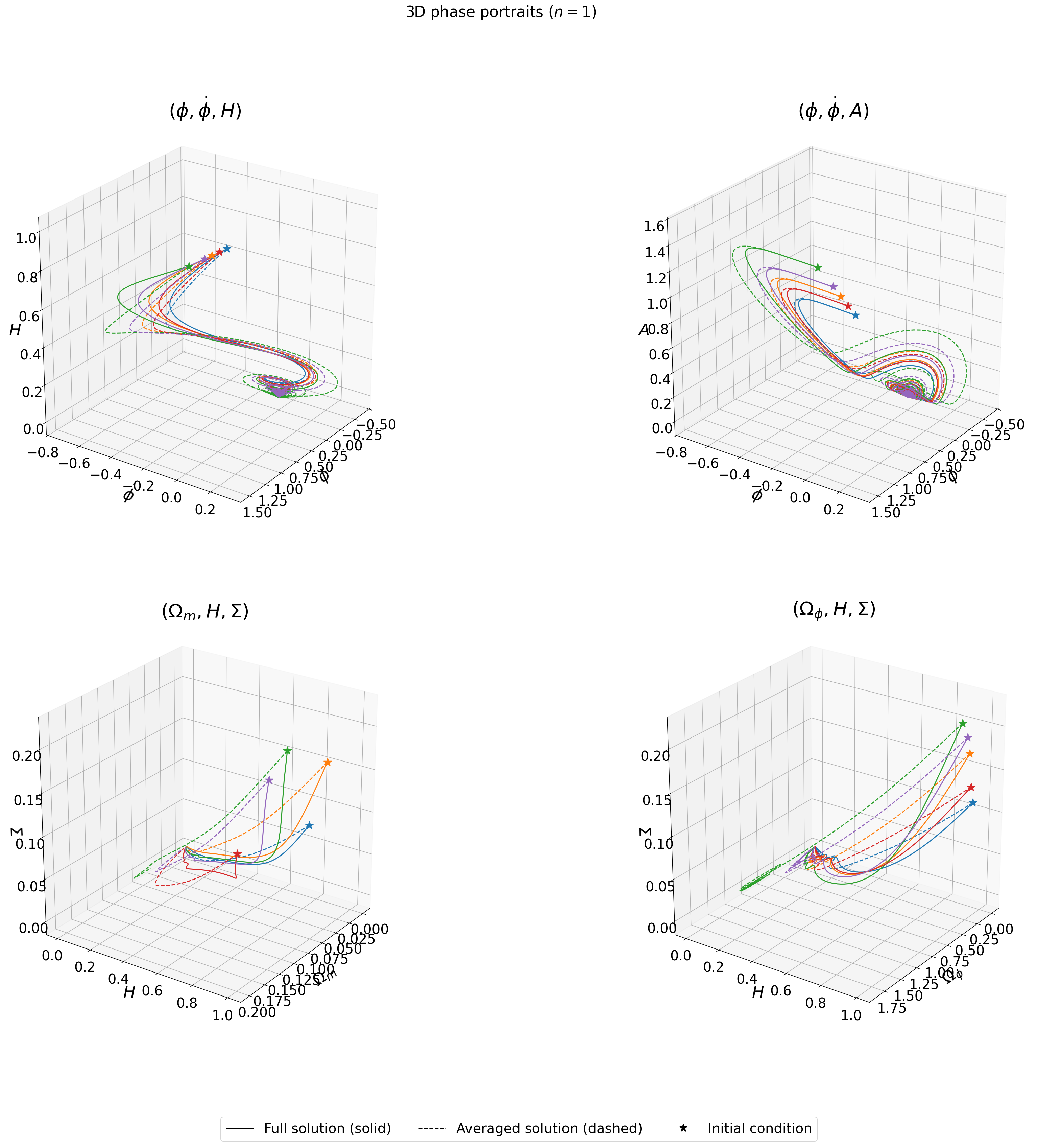}
    \caption{3D phase--space portraits for $n=1$.  Trajectories in $(\phi,\dot{\phi},H)$ and $(\phi,\dot{\phi},A)$ confirm that the averaged system provides a reliable autonomous approximation to the long-term dynamics. Initial conditions are taken from Table \ref{tb_1}.}
    \label{fig:phase3d_n1}
\end{figure}

\begin{figure}[ht!]
    \centering
    \includegraphics[width=0.8\textwidth]{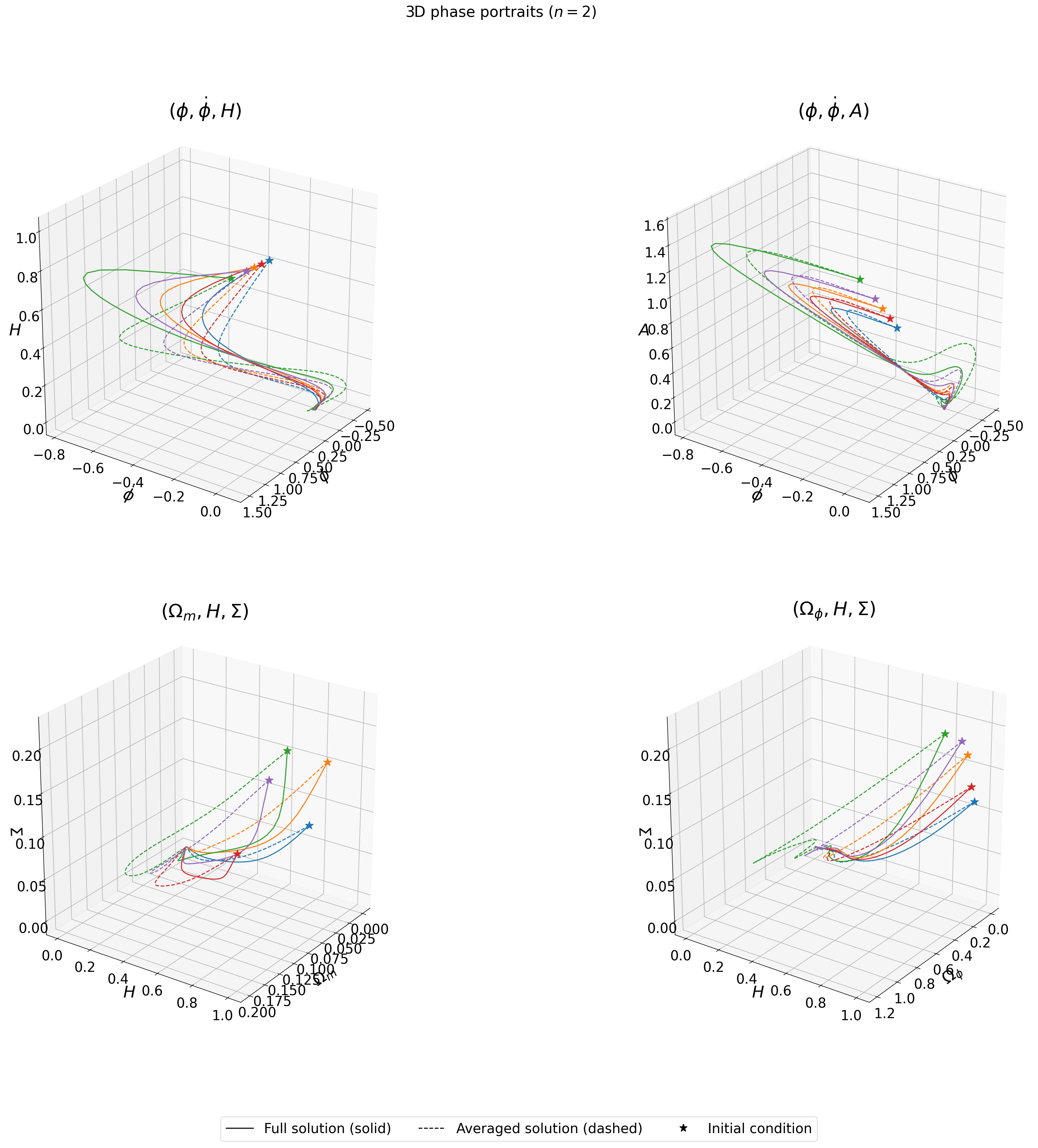}
    \caption{3D phase--space portraits for $n=2$.  Trajectories in $(\phi,\dot{\phi},H)$ and $(\phi,\dot{\phi},A)$ confirm that the averaged system provides a reliable autonomous approximation to the long-term dynamics. Initial conditions are taken from Table \ref{tb_1}.}
    \label{fig:phase3d_n2}
\end{figure}

The absence of closed periodic structures in the 2D and 3D portraits is a direct consequence of the cosmological dynamics. For $n=1$, the full system exhibits spirals in the $(\phi,\dot{\phi})$ plane, reflecting damped oscillations of the scalar field. These spirals collapse inward due to Hubble friction, which continuously drains energy from the field and prevents the trajectories from forming closed loops. The system is non-autonomous since $H(t)$ itself evolves, breaking time-translation invariance and eliminating the possibility of exact periodic orbits. By construction, the averaged amplitude--phase system smooths out the fast oscillations in $\theta$, so its trajectories cannot display periodicity either. As illustrated in the 2D and 3D portraits, both the full and averaged systems exhibit slow drifts toward attractors rather than repeating cycles, confirming that the absence of periodic structures is a physical feature of cosmological damping rather than a numerical artifact.

\subsection{Dynamics of the averaged system in the coordinates $(\overline{\Sigma},\overline{\Omega}_m,\overline{\Omega}_\phi)$ }

To explore the validity of the averaging procedure, we compare trajectories of three dynamical systems: 
the full amplitude--phase equations, the averaged amplitude--phase system, and the reduced averaged system 
in the compact coordinates $(\overline{\Sigma},\overline{\Omega}_m,\overline{\Omega}_\phi)$. 
The initial conditions are chosen from representative sets (see Table \ref{tb_1}), and each set is plotted with a distinct color. 
Solid lines denote the full system, dashed lines the averaged system, and dashed--dotted lines the reduced averaged system. Stars mark initial conditions taken from Table \ref{tb_1}. This visual comparison highlights how the averaged and reduced systems capture the coarse-grained dynamics of the full trajectories, filtering out fast oscillations while preserving the long-term flow toward attractors.  

For $n=1$ and $n=2$ with $\gamma=1$, the Kasner points are unstable sources, the matter and scalar FLRW points are saddles, 
and the curvature point $\mathcal{K}$ is the global attractor. The plots confirm this classification: 
trajectories originate near anisotropic or mixed states, pass close to matter or scalar saddles, 
and eventually converge to the curvature sink. The agreement between the three systems demonstrates 
that the averaging procedure faithfully reproduces the qualitative dynamics of the full equations.

\begin{figure}[ht!]
\centering
\includegraphics[width=\textwidth]{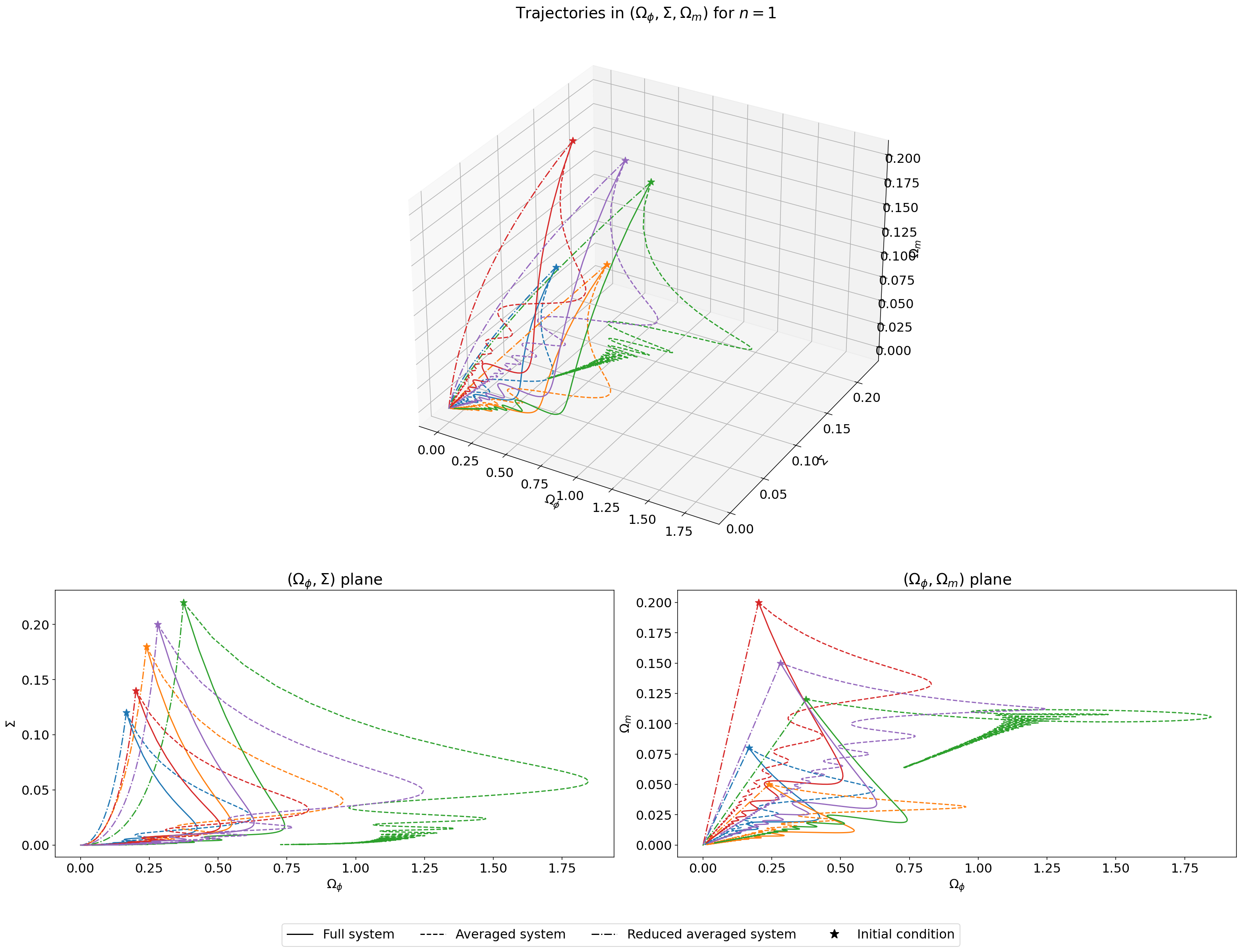}
\caption{Trajectories in $(\Omega_\phi,\Sigma,\Omega_m)$ for $n=1$, $\gamma=1$. 
Solid lines: full system; dashed lines: averaged system; dash--dotted lines: reduced averaged system. 
Stars mark initial conditions taken from Table \ref{tb_1}.}
\label{fig:n1_trajectories}
\end{figure}

\begin{figure}[ht!]
\centering
\includegraphics[width=\textwidth]{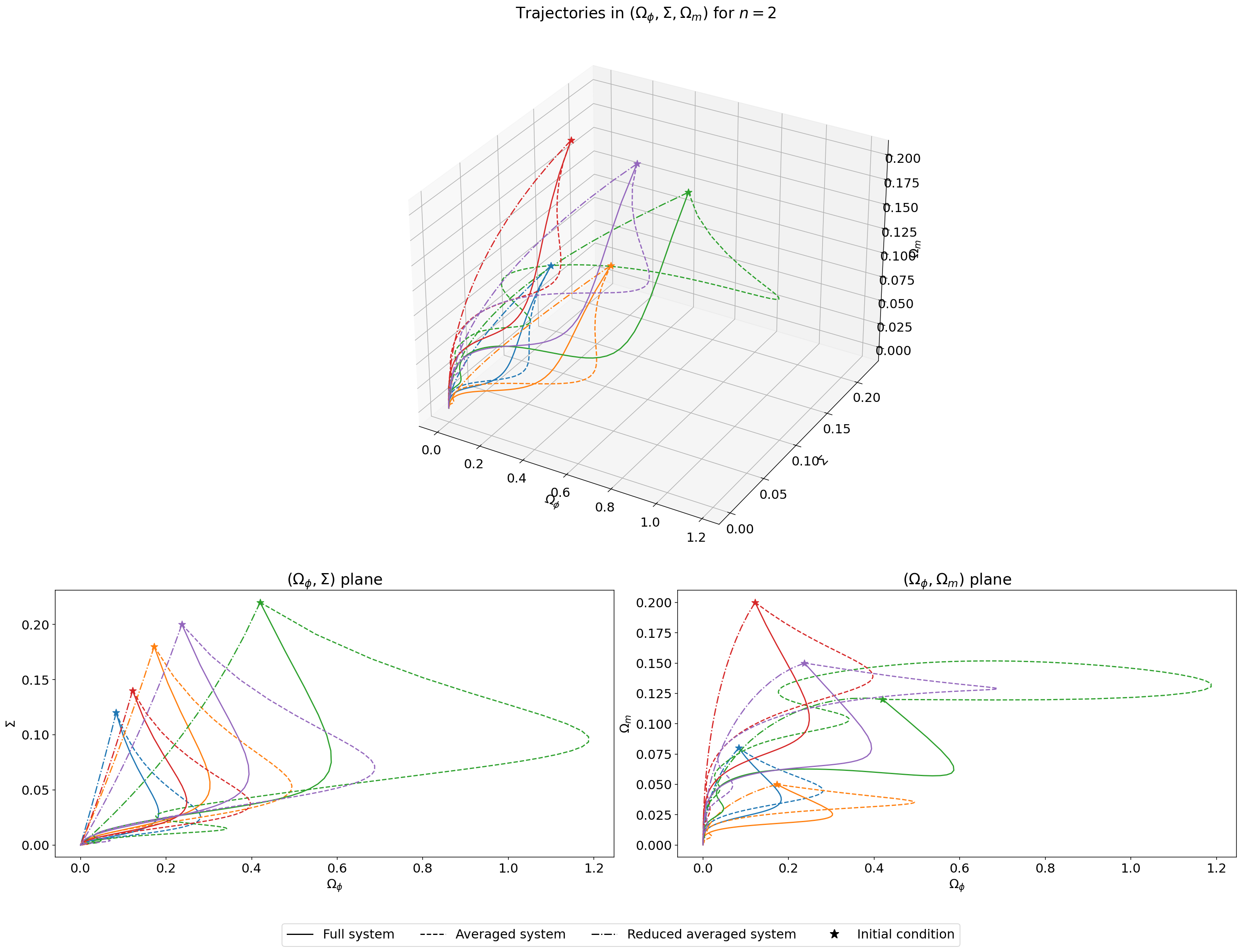}
\caption{Trajectories in $(\Omega_\phi,\Sigma,\Omega_m)$ for $n=2$, $\gamma=1$. 
Solid lines: full system; dashed lines: averaged system; dash--dotted lines: reduced averaged system. 
Stars mark initial conditions taken from Table \ref{tb_1}.}
\label{fig:n2_trajectories}
\end{figure}

\section{Conclusions}
\label{sect:VII}

We studied Bianchi V cosmologies within a compact dynamical systems framework. Motivated by the $\alpha$-attractor models in FLRW \cite{Alho:2023pkl}, the late-time dynamics were studied through monomial potentials $V(\phi)\sim \phi^{2n}$, which approximate the attractor models near their minima.

Two complementary averaging procedures were developed. 
Amplitude–phase averaging replaces rapidly oscillating trigonometric terms in the exact system by their cycle averages, yielding a regularized system that preserves equilibrium classification and clarifies the role of scalar oscillations as effective fluids. Coarse–grained averaging, valid when $\omega(A)\gg H$, derives explicit secular laws for amplitude and frequency, showing that $\dot{A}=-\tfrac{3H}{n+1}A$ and $\omega(A)\propto A^{n-1}$. This connects microscopic oscillations to macroscopic dilution laws, yielding a reduced averaged system.

The averaging theorem established in Appendix~\ref{avrg} provides the mathematical foundation for the analytic reductions used in this work. Under assumptions (H1)–(H4), solutions of the full Einstein–Klein–Gordon system and of the averaged system remain $O(\varepsilon)$‑close for times of order $1/\varepsilon$, ensuring that the averaged system faithfully reproduces the global attractor structure of the full dynamics. In particular, the scalar sector behaves as an effective barotropic fluid with index $\overline{\gamma}_\phi=\tfrac{2n}{n+1}$, up to $O(\varepsilon)$ corrections, and $\omega$‑limit sets coincide between the full and averaged flows. Thus, near the potential minimum, the oscillatory regime is consistently described by an effective fluid with $\overline{w}_\phi=\tfrac{n-1}{n+1}$ and dilution law $\overline{\rho_\phi}\propto a^{-6n/(n+1)}$, bridging microscopic scalar oscillations with macroscopic cosmological dynamics.

The matter solution acts as an attractor only when the scalar sector dilutes faster than matter, i.e.
\begin{equation}
\frac{6n}{n+1} > 3\gamma,
\end{equation}
so that $\overline{\rho_\phi}$ decays more rapidly than $\rho_m\propto a^{-3\gamma}$. For $\gamma=1$, this condition holds for $n>1$, while for $n=1$ the scalar behaves as dust and curvature eventually dominates. The qualitative behaviour of the scalar sector and the corresponding attractor structure are summarized in Table~\ref{tab:summary}.

\begin{table}[ht]
\centering
\begin{tabular}{|c|c|c|c|}
\hline
$n$ & Effective equation of state $\overline{w}_\phi$ & Dilution law $\overline{\rho_\phi}$ & Late-time attractor \\
\hline
1 & $0$ (dust-like) & $a^{-3}$ & Curvature sink $\mathcal{K}$ after matter saddle \\
2 & $\tfrac{1}{3}$ (radiation-like) & $a^{-4}$ & Curvature sink $\mathcal{K}$ after radiation saddle \\
\hline
\end{tabular}
\caption{Summary of effective scalar behaviour and attractor structure for monomial potentials $V(\phi)\sim \phi^{2n}$ with $\gamma=1$.}
\label{tab:summary}
\end{table}

The reduced averaged system admits five generic isolated equilibria: Kasner vacua $\mathcal{K}_0^\pm$, the matter FLRW point $\mathcal{F}$, the scalar FLRW point $\mathcal{S}$, and the curvature Milne-type point $\mathcal{K}$, together with special families for tuned $(n,\gamma)$. We find that $\mathcal{K}_0^\pm$ are always sources, $\mathcal{F}$ is generically a saddle but can act as a sink for $\gamma<\min\{\tfrac{2n}{n+1},\tfrac{2}{3}\}$, $\mathcal{S}$ is a sink if $0<n<\tfrac{1}{2}$ and $\tfrac{2n}{n+1}<\gamma\leq 2$, while $\mathcal{K}$ becomes a sink whenever $\gamma>\tfrac{2}{3}$ and $n>\tfrac{1}{2}$. These results demonstrate that isotropic FLRW $\alpha$-attractor models extend naturally to anisotropic Bianchi~V cosmologies: inflationary attractors remain robust, while the Milne-type curvature solution emerges as the late-time state.

In conclusion, amplitude–phase averaging provides a detailed analytic reduction of the oscillatory dynamics, while coarse–grained averaging yields explicit secular laws for amplitude and frequency. Their synthesis establishes a robust and efficient framework: the reduced averaged system captures the essential late–time behaviour of homogeneous Bianchi V cosmologies, regularizes boundaries, and clarifies invariant sets, while greatly improving computational efficiency. This unified approach demonstrates that the attractor structure familiar from FLRW $\alpha$–attractor models persists in anisotropic settings, ensuring that the late–time dynamics converge to the curvature sink $\mathcal{K}$ after transient matter or scalar phases. Future work will extend this analysis to the full evaluation of the $E$ and $T$ potentials, thereby connecting the local monomial approximation with the global attractor structure of $\alpha$–attractor cosmologies.

\acknowledgments{Funded by Agencia Nacional de Investigación y Desarrollo (ANID), Chile, through Proyecto Fondecyt Regular 2024, Folio 1240514, Etapa 2025. A.A. gratefully acknowledges the hospitality of Prof. Genly Leon and his research group at the Universidad Católica del Norte (UCN), Antofagasta, during his visit, where this work was initiated and largely carried out. We also extend our gratitude to the Vicerrectoría de Investigación y Desarrollo Tecnológico (VRIDT) of UCN for the scientific support provided through the Núcleo de Investigación en Geometría Diferencial y Aplicaciones, according to Resolution VRIDT No.~096/2022, and through the Núcleo de Investigación en Simetrías y la Estructura del Universo (NISEU), according to Resolution VRIDT No.~200/2025.}

\appendix

\section{Averaging: rigorous statement and proof}\label{avrg}

\subsection{Notation and functional setting}
Let $\mathbf{X}(t)\in\mathbb{R}^m$ denote the vector of slow Hubble‑normalized variables used in the main text (for example, $\Omega_\phi,\Sigma,\Omega_m,\ldots$). Let $\mathcal{K}\subset\mathbb{R}^m$ be the compact forward‑invariant set obtained by the compactification described in Sections II–III. Fix a norm $\|\cdot\|$ on $\mathbb{R}^m$ (equivalent to the Euclidean norm). For $r\in\mathbb{N}$ denote by $C^r(\mathcal{K})$ the Banach space of $r$-times continuously differentiable vector fields on $\mathcal{K}$ with the usual $C^r$-norm.

\subsection{Assumptions}

\begin{itemize}
\item \textbf{(H1) Regularity and compactness.} The full EKG vector field in the compact Hubble‑normalized variables is $C^{r}$ on an open neighbourhood of a compact forward‑invariant set $\mathcal{K}\subset\mathbb{R}^m$, with $r\ge 3$. All derivatives up to order $r$ are uniformly bounded there; denote by $M_r>0$ a uniform bound for the $C^r$-norm.

\item \textbf{(H2) Periodic family and modulation.} There exists an amplitude interval $I=[A_{\min},A_{\max}]$ with $0<A_{\min}<A_{\max}<\infty$ such that for each $A\in I$ a $C^{r-1}$ family of $2\pi$-periodic profiles $\Phi(\theta;A)$ solves the frozen oscillator at amplitude $A$. The fast phase and amplitude satisfy
\begin{equation}
\dot\theta=\omega(A)+\rho(A,\mathbf{X},t),\quad \omega(A)=\mu^{n}A^{\,n-1}\left(\frac{(2n)!}{2^{2n-1}(n!)^2}\right)^{1/2},
\end{equation}
with the explicit bound $|\rho(A,\mathbf{X},t)|\le C_\rho H(t)$, and
\begin{equation}
\dot A=-\frac{3H}{n+1}A+\delta_A(A,\mathbf{X},t),
\end{equation}
with $|\delta_A(A,\mathbf{X},t)|\le C_A\frac{H(t)^2}{\omega(A)}$. The constants $C_\rho,C_A>0$ depend only on $M_r$ and $I$.

\item \textbf{(H3) Scale separation.} Define
\begin{equation}
\varepsilon:=\sup_{\substack{\mathbf{X}\in\mathcal{K}\\A\in I}}\frac{H}{\omega(A)}.
\end{equation}
There exists $\varepsilon_0>0$ such that $0<\varepsilon\le\varepsilon_0\ll1$. In particular $\inf_{A\in I}\omega(A)>0$.

\item \textbf{(H4) Homological solvability.} For all $\mathbf{Y}\in\mathcal{K}$ and $A\in I$ the operator $\omega(A)\partial_\theta$ is invertible on zero‑mean $2\pi$-periodic functions and the small‑divisor denominators arising in the averaging construction are uniformly bounded below by a positive constant proportional to $\inf_{A\in I}\omega(A)$.
\end{itemize}

\subsection{Auxiliary results}

\begin{lemma}[Periodic family]\label{lem:periodic_family}
Under (H1) and for the monomial potential, there exists an amplitude interval $I$ and a $C^{r-1}$ family of $2\pi$-periodic profiles $\Phi(\theta;A)$ solving the frozen oscillator at fixed amplitude $A$. The frequency $\omega(A)=\mu^{n}A^{\,n-1}\left(\frac{(2n)!}{2^{2n-1}(n!)^2}\right)^{1/2}$ is $C^{r-1}$ on $I$ and $\inf_{A\in I}\omega(A)>0$.
\end{lemma}

\begin{proof}
The frozen scalar equation at fixed energy (or amplitude) defines a nonlinear oscillator with a nondegenerate minimum. Standard existence and smooth dependence results for periodic orbits (implicit function theorem applied to the Poincaré map or Lyapunov–Schmidt reduction) yield the family $\Phi(\theta;A)$ and smoothness in $A$. The stated lower bound on $\omega(A)$ follows by choosing $I$ away from zero amplitude.
\end{proof}

\begin{lemma}[Amplitude–phase modulation]\label{lem:ampphase_explicit}
Under (H1)–(H3), there exist constants $C_\rho,C_A>0$ and $T>0$ such that every solution with initial data in a neighbourhood of $\mathcal{K}$ admits
\begin{equation}
\phi(t)=A(t)\,\Phi\big(\theta(t);A(t)\big),
\end{equation}
with
\begin{equation}
\dot\theta=\omega(A)+\rho(A,\mathbf{X},t),\quad
\dot A=-\frac{3H}{n+1}A+\delta_A(A,\mathbf{X},t),
\end{equation}
and the uniform bounds
\begin{equation}
|\rho(A,\mathbf{X},t)|\le C_\rho H(t),\quad
|\delta_A(A,\mathbf{X},t)|\le C_A\frac{H(t)^2}{\omega(A)}
\end{equation}
for all $t\in[0,T]$, $\mathbf{X}\in\mathcal{K}$, and $A\in I$.
\end{lemma}

\begin{proof}
Project the full scalar equation onto the tangent and normal directions of the periodic family $\Phi(\theta;A)$ using a phase normalization condition. The coupling terms are proportional to $H$ and its derivatives; uniform $C^r$-bounds and compactness of $\mathcal{K}$ produce the stated constants and inequalities. This is the standard derivation of modulation in averaging theory.
\end{proof}

\begin{lemma}[Virial relation with explicit error]\label{lem:virial_explicit}
Under (H1)–(H3), for the normalized kinetic variable $x$ and potential contribution $y$ associated to $\phi$,
\begin{equation}
\frac{1}{2\pi}\int_0^{2\pi} x^2(\theta)\,d\theta = \frac{n}{n+1}\,\frac{1}{2\pi}\int_0^{2\pi}\big(x^2+y^2\big)\,d\theta + \mathcal{E}_1,
\end{equation}
with
\begin{equation}
|\mathcal{E}_1|\le C_1\varepsilon,
\end{equation}
for some $C_1>0$ depending only on $M_r$ and $I$.
\end{lemma}

\begin{proof}
For the frozen periodic orbit, the exact virial ratio $\langle x^2\rangle:\langle y^2\rangle=n:1$ holds. For the slowly modulated orbit, the difference between the time average and the frozen average is controlled by the modulation remainders from Lemma~\ref{lem:ampphase_explicit}; integrating yields the stated $O(\varepsilon)$ bound.
\end{proof}

\begin{proposition}[Near‑identity transform and remainder]\label{prop:nearid_explicit}
Under (H1)–(H4) there exists a near‑identity $C^{r-2}$ change of variables
\begin{equation}
\mathbf{X}=\mathbf{Y}+\varepsilon\,\mathbf{U}(\mathbf{Y},\theta,\varepsilon),
\end{equation}
with $\mathbf{U}$ $2\pi$-periodic in $\theta$ and uniformly bounded on $\mathcal{K}\times\mathbb{S}^1$, such that the transformed system for $\mathbf{Y}$ reads
\begin{equation}
\dot{\mathbf{Y}}=F_0(\mathbf{Y})+\mathcal{R}(\mathbf{Y},\theta,\varepsilon),
\end{equation}
where $F_0$ is the averaged vector field and
\begin{equation}
\sup_{\mathbf{Y}\in\mathcal{K},\;\theta\in\mathbb{S}^1}\|\mathcal{R}(\mathbf{Y},\theta,\varepsilon)\|\le C_2\varepsilon^2.
\end{equation}
The correction $\mathbf{U}$ solves the homological equation
\begin{equation}
\omega(\mathbf{Y})\partial_\theta \mathbf{U}(\mathbf{Y},\theta) = \widetilde G(\mathbf{Y},\theta),
\end{equation}
with $\widetilde G$ the zero‑mean oscillatory part of the original vector field; solvability follows from (H4).
\end{proposition}

\begin{proof}
Write the system on the slow time $\tau=\varepsilon t$, decompose the right‑hand side into averaged and oscillatory parts, and solve the homological equation for the periodic correction $\mathbf{U}$. Invertibility of $\omega(\mathbf{Y})\partial_\theta$ on zero‑mean functions (assured by (H4)) yields a bounded periodic $\mathbf{U}$. Substitution produces a remainder of order $\varepsilon^2$; regularity and uniformity follow from (H1) and compactness.
\end{proof}

\subsection{Main theorem}

\begin{thm}[Averaging for monomial potentials]\label{thm:rigorous_full}
Assume (H1)–(H4). Let $\mathbf{X}(t)$ be a solution of the full Einstein–Klein–Gordon system with initial data in $\mathcal{K}$, and let $\overline{\mathbf{X}}(t)$ be the solution of the averaged system (constructed using Lemma \ref{lem:virial_explicit}) with the same slow initial data. Then there exist constants $C>0$ and $T_0>0$ such that for all sufficiently small $\varepsilon$
\begin{equation}
\sup_{0\le t\le T_0/\varepsilon}\big\|\mathbf{X}(t)-\overline{\mathbf{X}}(t)\big\|\le C\varepsilon.
\end{equation}
Moreover, any $\omega$-limit point of $\mathbf{X}(t)$ as $t\to\infty$ (when it exists and lies in $\mathcal{K}$) is an $\omega$-limit point of $\overline{\mathbf{X}}(t)$, and vice versa. The averaged scalar sector has a barotropic index
$\overline{\gamma}_\phi=\frac{2n}{n+1}$,
up to errors $O(\varepsilon)$.
\end{thm}

\begin{proof}
Combine Lemmas \ref{lem:periodic_family}--\ref{lem:virial_explicit} and Proposition \ref{prop:nearid_explicit}.

\textbf{Averaged vector field.} Lemma \ref{lem:virial_explicit} replaces the fast kinetic term $x^2$ by its average $\tfrac{n}{n+1}\Omega_\phi$ up to $O(\varepsilon)$ errors, producing the averaged component \eqref{eq:eos}.

\textbf{Near‑identity transform.} Proposition \ref{prop:nearid_explicit} yields a transformed system with remainder $O(\varepsilon^2)$ uniformly on $\mathcal{K}\times\mathbb{S}^1$.

\textbf{Gronwall estimate.} Let $\mathbf{Y}(t)$ be the transformed full solution and $\overline{\mathbf{Y}}(t)$ the averaged solution in the same coordinates. Their difference satisfies a differential inequality with forcing $O(\varepsilon^2)$; Gronwall's lemma gives the uniform $O(\varepsilon)$ bound on $[0,T_0/\varepsilon]$. Transforming back yields the stated estimate.

\textbf{$\omega$-limit correspondence.} Uniform closeness on long times and compactness of $\mathcal{K}$ imply that accumulation points of one flow correspond to accumulation points of the other as $\varepsilon\to0$.

\textbf{Barotropic index.} The virial relation yields $\overline{w}_\phi=(n-1)/(n+1)$, hence $\overline{\gamma}_\phi=2n/(n+1)$, up to $O(\varepsilon)$.

\textbf{The constants $C_\rho,C_A,C_1,C_2$} depend only on the $C^r$-bounds $M_r$ of the vector field and on the amplitude interval $I$; they are uniform for all sufficiently small $\varepsilon$.
\end{proof}

The constructions above follow classical averaging theory. For a complete statement and explicit constants, see \cite{Sanders2007} and \cite{Bogoliubov1961}. In particular, \cite{Sanders2007} covers first‑order averaging and homological equations, while \cite{Bogoliubov1961} develops the KBM amplitude–phase construction. If the averaged attractor is hyperbolic, invariant‑manifold theory upgrades the finite long‑time $O(\varepsilon)$ closeness to global‑in‑time asymptotic convergence.

\bibliographystyle{unsrt}
\bibliography{refs.bib}

\end{document}